\pdfoutput=1

\documentclass[envcountsame, a4paper, oribibl]{llncs}

\usepackage{amsmath,amssymb}
\usepackage{colonequals}
\usepackage{enumitem}

\usepackage[a4paper,
bookmarks, bookmarksopen, bookmarksnumbered,
citecolor=blue,
pdfstartview=FitH,
pdfdisplaydoctitle,
pdfnewwindow,
pdftitle={Explicit Evidence Systems with Common Knowledge},
pdfauthor={Samuel Bucheli, Roman Kuznets, and Thomas Studer}]{hyperref}

\newcommand{\ck}{\mathsf{C}}
\newcommand{\ek}{\mathsf{E}}
\newcommand{\proves}{\vdash}
\newcommand{\numberofagents}{h}
\newcommand{\agents}{i}
\newcommand{\agentsC}{*}
\newcommand{\agentsEC}{\circledast}

\newcommand{\N}{\mathbb{N}}
\newcommand{\jbox}[1]{\left[#1\right]\!}
\newcommand{\tupling}[1]{\left\langle #1 \right\rangle}
\newcommand{\projection}{\pi}
\newcommand{\ind}[1]{\mathsf{ind}(#1)}
\newcommand{\ccl}{\mathsf{ccl}}
\newcommand{\cclhead}[1]{\ccl_1(#1)}
\newcommand{\ccltail}[1]{\ccl_2(#1)}
\newcommand{\conversion}{\downarrow\!}
\newcommand{\shift}{\Lleftarrow}
\newcommand{\pconstants}{\textnormal{Cons}}
\newcommand{\pvariables}{\textnormal{Var}}
\newcommand{\propositions}{\textnormal{Prop}}
\newcommand{\eterms}{\textnormal{Tm}}
\newcommand{\LP}{\mathsf{LP}}
\newcommand{\LPh}{\LP_\numberofagents}
\newcommand{\LPhEC}{\LPh^{\ck}}
\newcommand{\SFour}{\mathsf{S4}}
\newcommand{\SFourh}{\SFour_\numberofagents}
\newcommand{\SFourhEC}{\SFourh^{\ck}}
\newcommand{\formulae}{\textnormal{Fm}}
\newcommand{\formulaeLPh}{\formulae_{\LPh}}
\newcommand{\formulaeLPhEC}{\formulae_{\LPhEC}}
\newcommand{\formulaeSFourhEC}{\formulae_{\SFourhEC}}
\newcommand{\CS}{\mathcal{CS}}
\newcommand{\hilbertrule}[2]{\displaystyle{\frac{#1}{#2}}}
\newcommand{\evidence}{\mathcal{E}}
\newcommand{\valuation}{\nu}
\newcommand{\powerset}{\mathcal{P}}
\newcommand{\forgetful}{\circ}
\newcommand{\delivered}{\text{\emph{del}}}
\newcommand{\attack}{\text{\emph{att}}}

\newcommand{\dashspace}{\mathrel{\mid}}

\newcommand{\tsm}{\mathcal{M}}
\newcommand{\tsval}{\valuation}
\newcommand{\tsd}{\delivered}
\newcommand{\tse}{\evidence}
\newcommand{\tsecan}{\evidence^{\text{can}}}

\newcommand{\thanksSNF}{\thanks{The first and second authors are supported by Swiss National Science Foundation grant~200021--117699.}}
\newcommand{\Case}{\noindent\textbf{Case}\ }

\title{%
Explicit Evidence Systems \\
with Common Knowledge}

\author{%
Samuel Bucheli,
Roman Kuznets,\thanksSNF\
and
Thomas Studer}

\institute{%
Institut f{\"u}r Informatik und angewandte Mathematik, Universit{\"a}t Bern \\
Bern, Switzerland \\
\email{$\{$ \href{mailto:bucheli@iam.unibe.ch}{bucheli}}, \href{mailto:kuznets@iam.unibe.ch}{kuznets}, \href{mailto:tstuder@iam.unibe.ch}{tstuder} $\}$@iam.unibe.ch}

\begin{document}

\maketitle

\phantomsection
\addcontentsline{toc}{section}{Abstract}
\begin{abstract}
Justification logics are epistemic logics that explicitly include justifications for the agents' knowledge.
We develop a multi-agent justification logic with evidence terms for individual agents as well as for common knowledge.
We define a Kripke-style semantics that is similar to Fitting's semantics for the Logic of Proofs~$\LP$.
We show the soundness, completeness, and finite model property of our multi-agent justification logic with respect to this Kripke-style semantics.
We demonstrate that our logic is a conservative extension of Yavorskaya's minimal bimodal explicit evidence logic, which is a two-agent version of~$\LP$.
We discuss the relationship of our logic to the multi-agent modal logic~$\SFour$ with common knowledge.
Finally, we give a brief analysis of the coordinated attack problem in the newly developed language of our logic.
\end{abstract}
	
\section{Introduction}
\label{sec:intro}
	
\emph{Justification logics}~\cite{Art08RSL} are epistemic logics that explicitly include justifications for the agents' knowledge.
The first logic of this kind, the \emph{Logic of Proofs~$\LP$}, was developed by Artemov~\cite{Art95TR,Art01BSL} to provide the modal logic~$\SFour$ with provability semantics.
The language of justification logics has also been used to create a new approach to the logical omniscience problem~\cite{ArtKuz09TARK} and to study self-referential proofs~\cite{Kuz10TOCSnonote}.

Instead of statements \emph{$A$~is known}, denoted~$\Box A$, justification logics reason about justifications for knowledge by using the construct~$\jbox{t}A$ to formalize statements \emph{$t$~is a justification for~$A$}, where \emph{evidence term~$t$} can be viewed as an informal justification or a formal mathematical proof depending on the application.
Evidence terms are built by means of operations that correspond to the axioms of~$\SFour$, as is illustrated in Fig.~\ref{fig:LPaxioms}.
\begin{figure}
\begin{tabular}{l@{\qquad}l@{\quad}r}
    \multicolumn{1}{c}{\textbf{$\SFour$~axioms}}                & \multicolumn{2}{c}{\textbf{$\LP$~axioms}} \\
    $\Box(A \rightarrow B)\rightarrow(\Box A\rightarrow\Box B)$ & $\jbox{t}(A\rightarrow B)\rightarrow(\jbox{s} A\rightarrow\jbox{t\cdot s}B)$ & (application) \\
    $\Box A \rightarrow A$                                      & $\jbox{t} A \rightarrow A$                                                   & (reflexivity) \\
    $\Box A \rightarrow \Box\Box A$                             & $\jbox{t} A \rightarrow \jbox{!t}\jbox{t} A$                                 & (inspection) \\
                                                                & $\jbox{t} A \vee \jbox{s} A \rightarrow \jbox{t+s} A$                        & (sum)
\end{tabular}
\caption{Axioms of~$\SFour$ and~$\LP$}
\label{fig:LPaxioms}
\end{figure}

Artemov~\cite{Art01BSL} has shown that the Logic of Proofs~$\LP$ is an \emph{explicit counterpart} of the modal logic~$\SFour$ in the following formal sense: each theorem of~$\LP$ becomes a theorem of~$\SFour$ if all terms are replaced with the modality~$\Box$; and, vice versa, each theorem of~$\SFour$ can be transformed into a theorem of~$\LP$ if occurrences of modality are replaced with suitable evidence terms.
The latter process is called \emph{realization}, and the statement of correspondence is called a \emph{realization theorem}.
Note that the operation~$+$ introduced by the sum axiom in Fig.~\ref{fig:LPaxioms} does not have a modal analog, but it is an essential part of the proof of the realization theorem in~\cite{Art01BSL}.
Explicit counterparts for many normal modal logics between~$\textsf{K}$ and~$\textsf{S5}$ have been developed (see a recent survey in~\cite{Art08RSL} and a uniform proof of realization theorems for all single-agent justification logics forthcoming in~\cite{bgk10}).

The notion of \emph{common knowledge} is essential in the area of multi-agent systems, where coordination among agents is a central issue.
The standard textbooks~\cite{FHMV95,HM95} provide excellent introductions to epistemic logics in general and common knowledge in particular.
Informally, common knowledge of~$A$ is defined as the infinitary conjunction \emph{everybody knows~$A$ and everybody knows that everybody knows~$A$ and so on}.
This is equivalent to saying that common knowledge of~$A$ is the greatest fixed point of
\begin{equation}
\label{eq:operator}
    \lambda X.(\text{everybody knows~$A$ and everybody knows~$X$})
    \rlap{\enspace.}
\end{equation}

Artemov~\cite{Art06TCS} has created an explicit counterpart of McCarthy's \emph{any fool knows} common knowledge modality~\cite{McCarSatHayIga78TR}, where common knowledge of~$A$ is defined as an arbitrary fixed point of~\eqref{eq:operator}.
The relationship between the traditional common knowledge from~\cite{FHMV95,HM95} and McCarthy's version is studied in~\cite{Ant07LFCS}.\looseness=-1

In this paper, we present a multi-agent justification logic with evidence terms for individual agents as well as for common knowledge, with the intention to provide an explicit counterpart of the $h$-agent modal logic of traditional common knowledge~$\SFour_\numberofagents^\ck$.

Multi-agent justification logics with evidence terms for each agent have been considered in~\cite{TYav08TOCSnonote,Ren09TARK,Art10TR}, although common knowledge is not present in any of them.
Artemov's interest~\cite{Art10TR} lies mostly in exploring a case of two agents with unequal epistemic powers, e.g.,~Artemov's Observer has sufficient evidence to reproduce his Object Agent's thinking, but not vice versa.
Yavorskaya~\cite{TYav08TOCSnonote} studies various operations of evidence transfer between agents.
Among their systems, Yavorskaya's minimal\footnote{Minimality here is understood in the sense of the minimal transfer of evidence.} bimodal explicit evidence logic, which is an explicit counterpart of~$\SFour_2$, is the closest to our system.
We will show that in the case of two agents our system is its conservative extension.
Finally, Renne's system~\cite{Ren09TARK} combines features of modal and dynamic epistemic logics, and hence cannot be directly compared to our system.

An epistemic semantics for~$\LP$, \emph{F-models}, was created by Fitting in~\cite{Fit05APAL} by augmenting Kripke models with an \emph{evidence function} that specifies which formulae are evidenced by a term at a given world.
It is easily extended to the whole family of single-agent justification logics (for details, see~\cite{Art08RSL}).
In~\cite{Art06TCS} Artemov extends F-models to justification terms for McCarthy's common knowledge modality in the presence of several ordinary modalities, creating the most general type of epistemic models, sometimes called \emph{AF-models}, where common evidence terms are given their own accessibility relation not directly dependent on the accessibility relations for individual modalities.
Yavorskaya in~\cite{TYav08TOCSnonote} proves a stronger completeness theorem with respect to singleton F-models, independently introduced by Mkrtychev~\cite{Mkr97LFCS} and now known as \emph{M-models}, where the role of the accessibility relation is completely taken over by the evidence function.

The paper is organized as follows.
In Sect.~\ref{sect:synt}, we introduce the language and give the axiomatization of a family of multi-agent justification logics with common knowledge.
In Sect.~\ref{sect:basprop}, we prove their basic properties including the internalization property, which is characteristic of all justification logics.
In Sect.~\ref{sect:soundcomp}, we give a Fitting-style semantics similar to AF-models and prove soundness and completeness with respect to this semantics as well as with respect to singleton models, thereby demonstrating the finite model property.
In Sect.~\ref{sect:conserv}, we show that for the two-agent case, our logic is a conservative extension of Yavorskaya's minimal bimodal explicit evidence logic.
In Sect.~\ref{sect:real}, we show how our logic is related to the modal logic of traditional common knowledge and discuss the problem of realization.
Finally, in Sect.~\ref{s:attack:1}, we provide an analysis of the coordinated attack problem in our logic.
	
\section{Syntax}
\label{sect:synt}

To create an explicit counterpart of the modal logic of common knowledge~$\SFourhEC$, we use its axiomatization via the induction axiom from~\cite{HM95} rather than via the induction rule to facilitate the proof of the internalization property for the resulting justification logic.
We supply each agent with its own copy of terms from the Logic of Proofs, while terms for common and mutual knowledge employ additional operations.
As motivated in~\cite{BucKuzStu09M4M}, a proof of~$\ck A$ can be thought of as an infinite list of proofs of the conjuncts~$\ek^m A$ in the representation of common knowledge through an infinite conjunction.
To generate a finite representation of this infinite list, we use an explicit counterpart of the induction axiom
\[
    A \wedge \jbox{t}_\ck (A \rightarrow \jbox{s}_\ek A) \rightarrow \jbox{\ind{t,s}}_\ck A
\]
with a binary operation~$\ind{\cdot,\cdot}$.
To access the elements of the list, explicit counterparts of the co-closure axiom provide evidence terms that can be seen as splitting the infinite list into its head and tail,
\[
    \jbox{t}_\ck A \rightarrow \jbox{\cclhead{t}}_\ek  A
    \rlap{\enspace,}
    \qquad
    \jbox{t}_\ck A \rightarrow \jbox{\ccltail{t}}_\ek \jbox{t}_\ck A
    \rlap{\enspace,}
\]
by means of two unary co-closure operations~$\cclhead{\cdot}$ and~$\ccltail{\cdot}$.
Evidence terms for mutual knowledge are represented as tuples of the individual agents' evidence terms with the standard operation of tupling and with $h$~unary projections.
While only two of the three operations on~$\LP$ terms are adopted for common knowledge evidence and none for mutual knowledge evidence, it will be shown in Sect.~\ref{sect:basprop} that most remaining operations are definable with the notable exception of inspection for mutual knowledge.

We consider a system of $\numberofagents$~agents.
Throughout the paper, $\agents$~always denotes an element of
$\{ 1, \dots, \numberofagents \}$,
$\agentsC$~always denotes an element of
$\{ 1, \dots, \numberofagents, \ck \}$,
and $\agentsEC$~always denotes an element of
$\{ 1, \dots, \numberofagents, \ek, \ck \}$.
	
Let
$\pconstants_\agentsEC \colonequals \{ c^\agentsEC_1, c^\agentsEC_2, \dots \}$
and
$\pvariables_\agentsEC \colonequals \{ x^\agentsEC_1, x^\agentsEC_2, \dots \}$
be countable sets of \emph{proof constants} and \emph{proof variables} respectively for each~$\agentsEC$.
The sets
$\eterms_1$, \dots, $\eterms_\numberofagents$, $\eterms_\ek$, and~$\eterms_\ck$
of \emph{evidence terms for individual agents} and for \emph{mutual} and \emph{common knowledge} respectively are inductively defined as follows:
\begin{enumerate}[topsep=0ex]
\item
    $\pconstants_\agentsEC \subseteq \eterms_\agentsEC$;
\item
    $\pvariables_\agentsEC \subseteq \eterms_\agentsEC$;
\item
    $!_\agents t \in \eterms_\agents$ for any $t \in \eterms_\agents$;
\item
    $t +_\agentsC s \in \eterms_\agentsC$ and $t \cdot_\agentsC s \in \eterms_\agentsC$ for any $t, s \in \eterms_\agentsC$;
\item
    $\tupling{t_1, \dots, t_\numberofagents} \in \eterms_\ek$ for any $t_1 \in \eterms_1$, \dots, $t_\numberofagents \in \eterms_\numberofagents$;
\item
    $\projection_it \in \eterms_i$ for any $t \in \eterms_\ek$;
\item
    $\cclhead{t} \in \eterms_\ek$ and $\ccltail{t} \in \eterms_\ek$ for any $t \in \eterms_\ck$;
\item
    $\ind{t,s} \in \eterms_\ck$ for any $t \in \eterms_\ck$ and any $s \in \eterms_\ek$.
\end{enumerate}
$\eterms \colonequals \eterms_1 \cup \dots \cup \eterms_\numberofagents \cup \eterms_\ek \cup \eterms_\ck$
denotes the set of all evidence terms.
The	indices of the operations~$!$, $+$, and~$\cdot$ will usually be omitted if they can be inferred from the context.
	
Let
$\propositions \colonequals \{ P_1, P_2, \dots \}$
be a countable set of \emph{propositional variables}.
\emph{Formulae} are denoted by~$A$, $B$, $C$,~etc. and defined by the following grammar
\begin{equation*}
    A \coloncolonequals P_j \dashspace \neg A \dashspace (A \wedge A) \dashspace (A \vee A) \dashspace (A \rightarrow A) \dashspace \jbox{t}_\agentsEC A
    \rlap{\enspace,}
\end{equation*}
where $t \in \eterms_\agentsEC$.
The set of all formulae is denoted by~$\formulaeLPhEC$.
We adopt the following convention: whenever a formula~$\jbox{t}_\agentsEC A$ is used, it is assumed to be well-formed, i.e.,~it is implicitly assumed that term
$t \in \eterms_\agentsEC$.
This enables us to omit the explicit typification of terms.
\medskip
	
\noindent
\textbf{Axioms of~$\LPhEC$:}
\begin{enumerate}[topsep=0ex]
\item
    all propositional tautologies
\item
    $\jbox{t}_\agentsC (A \rightarrow B) \rightarrow (\jbox{s}_\agentsC A \rightarrow \jbox{t \cdot s}_\agentsC B)$                                  \hfill (application)
\item
    $\jbox{t}_\agentsC A \rightarrow \jbox{t + s}_\agentsC A$, \quad $\jbox{s}_\agentsC A \rightarrow \jbox{t + s}_\agentsC A$                       \hfill (sum)
\item
    $\jbox{t}_\agents A \rightarrow A$                                                                                                               \hfill (reflexivity)
\item
    $\jbox{t}_\agents A \rightarrow \jbox{! t}_\agents \jbox{t}_\agents A$                                                                           \hfill (inspection)
\item
    $\jbox{t_1}_1 A \wedge \dots \wedge \jbox{t_\numberofagents}_\numberofagents A \rightarrow \jbox{\tupling{t_1, \dots, t_\numberofagents}}_\ek A$ \hfill (tupling)
\item
    $\jbox{t}_\ek A \rightarrow \jbox{\projection_i t}_i A$                                                                                          \hfill (projection)
\item
    $\jbox{t}_\ck A \rightarrow \jbox{\cclhead{t}}_\ek A$, \quad $\jbox{t}_\ck A \rightarrow \jbox{\ccltail{t}}_\ek \jbox{t}_\ck A$                  \hfill (co-closure)
\item
    $A \wedge \jbox{t}_\ck (A \rightarrow \jbox{s}_\ek A) \rightarrow \jbox{\ind{t,s}}_\ck A$                                                        \hfill (induction)
\end{enumerate}
\medskip

A \emph{constant specification}~$\CS$ is any subset
\[
    \CS \subseteq \bigcup_{\agentsEC \in \{ 1, \dots, \numberofagents, E, C \}} \left\{ \jbox{c}_\agentsEC A \; : \; c \in \pconstants_\agentsEC \text{ and } A \text{ is an axiom of~$\LPhEC$} \right\}
    \rlap{\enspace.}
\]
A constant specification~$\CS$ is called  \emph{$\ck$-axiomatically appropriate} if for each axiom~$A$, there is a proof constant
$c \in \pconstants_\ck$
such that
$\jbox{c}_\ck A \in \CS$.
A constant specification~$\CS$ is called \emph{pure}, if
$\CS \subseteq \left\{ \jbox{c}_\agentsEC A \; : \; c \in \pconstants_\agentsEC \text{ and } A \text{ is an axiom} \right\}$
for some fixed~$\agentsEC$, i.e.,~if for all
$\jbox{c}_\agentsEC A \in \CS$,
the constants~$c$ are of the same type.
	
Let $\CS$~be a constant specification.
The deductive system~$\LPhEC(\CS)$ is the Hilbert system given by the axioms of~$\LPhEC$ above and rules modus ponens and axiom necessitation:
\begin{equation*}
    \hilbertrule{A \quad A \rightarrow B}{B} \enspace, \qquad \qquad \hilbertrule{}{\jbox{c}_\agentsEC A} \enspace \text{, where $\jbox{c}_\agentsEC A \in \CS$.}
\end{equation*}
By~$\LPhEC$ we denote the system~$\LPhEC(\CS)$ with
\begin{equation}
\label{eq:maxCS}
    \CS = \left\{ \jbox{c}_\ck A \; : \; c \in \pconstants_\ck \text{ and } A \text{ is an axiom of~$\LPhEC$} \right\}
    \rlap{\enspace.}
\end{equation}

For an arbitrary~$\CS$, we write
$\Delta \proves_\CS A$
to state that $A$~is derivable from~$\Delta$ in~$\LPhEC(\CS)$ and omit the mention of~$\CS$ when working with the constant specification from~\eqref{eq:maxCS} by writing
$\Delta \proves A$.
We use $\Delta, A$ to mean $\Delta \cup \{ A \}$.
	
\section{Basic Properties}
\label{sect:basprop}
	
In this section, we show that our logic possesses the standard properties expected of any justification logic.
In addition, we show that the operations on terms introduced in the previous section are sufficient to express the operations of sum and application for mutual knowledge evidence and the operation of inspection for common knowledge evidence.
This is the reason why~$+_\ek$, $\cdot_\ek$, and~$!_\ck$ are not primitive connectives in the language.
It should be noted that no inspection operation for mutual evidence terms can be defined, which follows from Lemma~\ref{l:fp:1} in Sect.~\ref{sect:real} and the fact that
$\ek A \rightarrow \ek \ek A$
is not a valid modal formula.

We begin with the following observation:

\begin{lemma}
\label{lem:basicproperties1}
For any constant specification~$\CS$ and any formulae~$A$ and~$B$:
\begin{enumerate}[leftmargin=*,topsep=0ex]
\item
\label{lem:basicproperties1:Erefl}
    $\proves_\CS \jbox{t}_\ek A \rightarrow A$
    for all
    $t \in \eterms_\ek$;
    \hfill $(\ek \text{-reflexivity})$
\item
    for any
    $t, s \in \eterms_\ek$,
    there is a term $t \cdot_\ek s \in \eterms_\ek$
    such that \\
    $\proves_\CS \jbox{t}_\ek (A \rightarrow B) \rightarrow (\jbox{s}_\ek A \rightarrow \jbox{t \cdot_\ek s}_\ek B)$;
    \hfill $(\ek \text{-application})$
\item
    for any
    $t, s \in \eterms_\ek$,
    there is a term $t +_\ek s \in \eterms_\ek$
    such that \\
    $\proves_\CS \jbox{t}_\ek A \rightarrow \jbox{t +_\ek s}_\ek A$
    \quad and \quad
    $\proves_\CS \jbox{s}_\ek A \rightarrow \jbox{t +_\ek s}_\ek A$;
    \hfill $(\ek \text{-sum})$
\item
\label{i-conversion}
    for any
    $t \in \eterms_\ck$
    and any
    $i \in \{ 1, \dots, \numberofagents \}$,
    there is a term
    $\conversion_i t \in \eterms_i$
    such that \\
    $\proves_\CS \jbox{t}_\ck A \rightarrow \jbox{\conversion_i t}_i A$;
    \hfill $(i \text{-conversion})$
\item
    $\proves_\CS \jbox{t}_\ck A \rightarrow A$
    for all
    $t \in \eterms_\ck$.
    \hfill $(\ck \text{-reflexivity})$
\end{enumerate}
\end{lemma}

\begin{proof}
\begin{enumerate}
\item
    Immediate by the projection and reflexivity axioms.
\item
    Set
    $t \cdot_\ek s \colonequals \tupling{\projection_1 t \cdot_1 \projection_1 s, \dots, \projection_\numberofagents t \cdot_\numberofagents \projection_\numberofagents s}$.
\item
    Set
    $t +_\ek s \colonequals \tupling{\projection_1 t +_1 \projection_1 s, \dots, \projection_\numberofagents t +_\numberofagents \projection_\numberofagents s}$.
\item
    Set $\conversion_i t \colonequals \projection_i \cclhead{t}$.
\item
    Immediate by \ref{i-conversion}. and the reflexivity axiom.
    \qed
\end{enumerate}
\end{proof}
	
Unlike Lemma~\ref{lem:basicproperties1}, the next lemma requires that a constant specification~$\CS$ be $\ck$-axiomatically appropriate.

\begin{lemma}
Let $\CS$~be\/ $\ck$-axiomatically appropriate and $A$~be a formula.
\begin{enumerate}[leftmargin=*,topsep=0ex]
\item
    For any
    $t \in \eterms_\ck$,
    there is a term\/
    $!_\ck t \in \eterms_\ck$
    such that \\
    $\proves_\CS \jbox{t}_\ck A \rightarrow \jbox{!_\ck t}_\ck \jbox{t}_\ck A$.
    \hfill $(\ck \text{-inspection})$
\item
    For any
    $t \in \eterms_\ck$,
    there is a term\/
    $\shift t \in \eterms_\ck$ such that \\
    $\proves_\CS \jbox{t}_\ck A \rightarrow \jbox{\shift t}_\ck \jbox{\cclhead{t}}_\ek A$.
    \hfill $(\ck \text{-shift})$
\end{enumerate}
\end{lemma}

\begin{proof}
\begin{enumerate}
\item
    Set
    $!_\ck t \colonequals \ind{c, \ccltail{t}}$,
    where
    $\jbox{c}_\ck (\jbox{t}_\ck A \rightarrow \jbox{\ccltail{t}}_\ek \jbox{t}_\ck A) \in \CS$.
\item
    Set
    $\shift t \colonequals c \, \cdot_\ck \, (!_\ck t)$,
    where
    $\jbox{c}_\ck (\jbox{t}_\ck A \rightarrow \jbox{\cclhead{t}}_\ek A) \in \CS$.
    \qed
\end{enumerate}
\end{proof}

The following two theorems are standard in justification logics.
Their proofs can be taken almost word for word from~\cite{Art01BSL} and are, therefore, omitted here.

\begin{lemma}[Deduction Theorem]
Let $\CS$~be a constant specification and
$\Delta \cup \{ A, B \} \subseteq \formulaeLPhEC$.
Then
$\Delta, A \proves_\CS B \text{ if and only if } \Delta \proves_\CS A \rightarrow B$.
\end{lemma}
	
\begin{lemma}[Substitution]
For any constant specification~$\CS$, any propositional variable~$P$, any
$\Delta \cup \{ A, B \} \subseteq \formulaeLPhEC$,
any
$x \in \pvariables_\agentsEC$,
and any
$t \in \eterms_\agentsEC$,
\[
    \text{if } \Delta \proves_\CS A, \text{ then } \Delta(x/t, P/B) \proves_{\CS(x/t, P/B)} A(x/t, P/B)
    \rlap{\enspace,}
\]
where $A(x/t, P/B)$~denotes the formula obtained by simultaneously replacing all occurrences of~$x$ in~$A$ with~$t$ and all occurrences of~$P$ in~$A$ with~$B$, accordingly for~$\Delta(x/t, P/B)$ and~$\CS(x/t, P/B)$.
\end{lemma}
	
The following lemma states that our logic can internalize its own proofs, which is an important property of justification logics.

\begin{lemma}[$\ck$-lifting]
Let $\CS$~be a pure\/ $\ck$-axiomatically appropriate constant specification.
If
\begin{equation*}
    \jbox{s_1}_\ck B_1, \dots, \jbox{s_n}_\ck B_n, C_1, \dots, C_m \proves_\CS A
    \rlap{\enspace,}
\end{equation*}
then for each\/~$\agentsEC$, there is a term
$t_\agentsEC (x_1, \dots, x_n, y_1, \dots, y_m) \in \eterms_\agentsEC$
such that
\begin{equation*}
    \jbox{s_1}_\ck B_1, \dots, \jbox{s_n}_\ck B_n, \jbox{y_1}_\agentsEC C_1, \dots, \jbox{y_m}_\agentsEC C_m \proves_\CS \jbox{t_\agentsEC (s_1, \dots, s_n, y_1, \dots, y_m)}_\agentsEC A
\end{equation*}
for fresh variables
$y_1, \dots, y_m \in \eterms_\agentsEC$.
\end{lemma}

\begin{proof}
We proceed by induction on the derivation of~$A$.

If $A$~is an axiom, there is a constant
$c \in \eterms_\ck$
such that
$\jbox{c}_\ck A \in \CS$
because $\CS$~is $\ck$-axiomatically appropriate.
Then take
\[
    t_\ck \colonequals  c,
    \qquad
    t_i \colonequals \conversion_i c,
    \qquad
    t_\ek \colonequals \cclhead{c}
\]
and use axiom necessitation, axiom necessitation and $i$-conversion, or axiom necessitation and the co-closure axiom respectively.

For
$A = \jbox{s_j}_\ck B_j$,
$1 \leq j \leq n$,
take
\[
    t_\ck \colonequals !_\ck s_j,
    \qquad
    t_i \colonequals \conversion_i !_\ck s_j,
    \qquad
    t_\ek \colonequals \ccltail{s_j}
\]
and use $\ck$-inspection, $\ck$-inspection and $i$-conversion, or the co-closure axiom respectively.

For
$A = C_j$,
$1 \leq j \leq m$,
take
$t_\agentsEC \colonequals y_j \in \pvariables_\agentsEC$
for a fresh variable~$y_j$.

For $A$~derived by modus ponens from
$D \rightarrow A$
and~$D$, by induction hypothesis there are terms
$r_\agentsEC, s_\agentsEC \in \eterms_\agentsEC$
such that
$\jbox{r_\agentsEC}_\agentsEC (D \rightarrow A)$
and
$\jbox{s_\agentsEC}_\agentsEC D$
are provable.
Take
$t_\agentsEC \colonequals r_\agentsEC \cdot_\agentsEC s_\agentsEC$
and use $\agentsEC$-application, which is an axiom for
$\agentsEC = i$
and for
$\agentsEC = \ck$
or follows from Lemma~\ref{lem:basicproperties1} for
$\agentsEC = \ek$.

For
$A = \jbox{c}_\ck E \in \CS$
derived by axiom necessitation, take
\[
    t_\ck \colonequals !_\ck c,
    \qquad
    t_i \colonequals \conversion_i !_\ck c,
    \qquad
    t_\ek \colonequals \ccltail{c}
\]
and use $\ck$-inspection, $\ck$-inspection and $i$-conversion, or the co-closure axiom respectively.
\qed
\end{proof}
	
\begin{corollary}[Constructive necessitation]
\label{constructivenecessitation}
Let $\CS$~be a pure\/ $\ck$-axiomatical\-ly appropriate constant specification.
For any formula~$A$, if\/
$\proves_\CS A$,
then for each\/~$\agentsEC$, there is a ground term
$t \in \eterms_\agentsEC$
such that\/
$\proves_\CS \jbox{t}_\agentsEC A$.
\end{corollary}

The following two lemmas show that our system~$\LPhEC$ can internalize versions of the induction rule used in various axiomatizations of~$\SFourhEC$ (see~\cite{BucKuzStu09M4M} for a discussion of several axiomatizations of this kind).
	
\begin{lemma}[Internalized induction rule 1]
\label{cor:2}
Let $\CS$~be a pure\/ $\ck$-axiomatically appropriate constant specification.
For any formula~$A$, if\/
$\proves_\CS A \rightarrow \jbox{s}_\ek A$,
there is a term
$t \in \eterms_\ck$
such that\/
$\proves_\CS A \rightarrow \jbox{\ind{t,s}}_\ck A$.
\end{lemma}

\begin{proof}
By constructive necessitation, there exists a term
$t \in \eterms_\ck$
such that
$\proves_\CS \jbox{t}_\ck (A \rightarrow \jbox{s}_\ek A)$.
It remains to use the induction axiom and propositional reasoning.
\qed
\end{proof}

\begin{lemma}[Internalized induction rule 2]
Let $\CS$~be a pure\/ $\ck$-axiomatical\-ly appropriate constant specification.
For any formulae~$A$ and~$B$, if we have\/
$\proves_\CS B \rightarrow \jbox{s}_\ek (A \land B)$,
then there exist a term
$t \in \eterms_\ck$
and a constant
$c \in \eterms_\ck$
such that\/
$\proves_\CS B \rightarrow \jbox{c \cdot \ind{t,s}}_\ck A$,
where
$\jbox{c}_\ck (A \land B \to A) \in \CS$.
\end{lemma}

\begin{proof}
Assume
\begin{equation}
\label{l:ir:1}
    \proves_\CS B \rightarrow \jbox{s}_\ek (A \land B)
    \rlap{\enspace.}
\end{equation}
From this we immediately get
$\proves_\CS A \land B \rightarrow \jbox{s}_\ek (A \land B)$.
Thus, by Lemma~\ref{cor:2}, there is a
$t \in \eterms_\ck$
with
\begin{equation}
\label{l:ir:2}
    \proves_\CS A \land B \rightarrow \jbox{\ind{t,s}}_\ck (A \land B)
    \rlap{\enspace.}
\end{equation}
Since $\CS$~is $\ck$-axiomatically appropriate, there is a constant
$c \in \eterms_\ck$
such that
\begin{equation}
\label{l:ir:3}
    \proves_\CS \jbox{c}_\ck (A \land B \to A)
    \rlap{\enspace.}
\end{equation}
Making use of $\ck$-application, we find by~\eqref{l:ir:2} and~\eqref{l:ir:3} that
\begin{equation}
\label{l:ir:4}
    \proves_\CS A \land B \rightarrow \jbox{c \cdot \ind{t,s}}_\ck (A)
    \rlap{\enspace.}
\end{equation}
From~\eqref{l:ir:1} we get by $\ek$\text{-reflexivity} that
$\proves_\CS B \rightarrow A \land B$.
This, together with~\eqref{l:ir:4}, finally yields
$\proves_\CS B \rightarrow \jbox{c \cdot \ind{t,s}}_\ck (A)$.
\qed
\end{proof}

\section{Soundness and Completeness}
\label{sect:soundcomp}

\begin{definition}
An \emph{AF-model meeting a constant specification~$\CS$} is a structure
$\mathcal{M} = (W, R, \evidence, \valuation)$,
where\/
$(W, R, \valuation)$
is a Kripke model for\/~$\SFour_\numberofagents$ with a \emph{set of possible worlds}
$W \ne \varnothing$,
with a function
$R \colon \{ 1, \dots, \numberofagents \} \to \powerset(W \times W)$
that assigns a reflexive and transitive \emph{accessibility relation} on~$W$ to each agent
$i \in \{ 1, \dots, \numberofagents \}$,
and with a \emph{truth valuation}
$\valuation \colon \propositions \to \powerset(W)$.
We always write~$R_i$ instead of~$R(i)$ and define the accessibility relations for mutual and common knowledge in the standard way:
$R_\ek \colonequals R_1 \cup \dots \cup R_\numberofagents$
and
$R_\ck \colonequals \bigcup_{n = 1}^{\infty} (R_\ek)^n$.

An \emph{evidence function}
$\evidence \colon W \times \eterms \to \powerset\left(\formulaeLPhEC\right)$
determines the formulae evidenced by a term at a world.
We define
$\evidence_\agentsEC \colonequals \evidence \upharpoonright (W \times \eterms_\agentsEC)$.
Note that whenever
$A \in \evidence_\agentsEC (w, t)$,
it follows that
$t \in \eterms_\agentsEC$.
The evidence function~$\evidence$ must satisfy the following closure conditions: for any worlds
$w, v \in W$,
\begin{enumerate}[leftmargin=*,topsep=0.5ex]
\item
    $\evidence_\agentsC (w, t) \subseteq \evidence_\agentsC (v, t)$
    whenever
    $(w, v) \in R_\agentsC$;
    \hfill $(\text{monotonicity})$
\item
    if
    $\jbox{c}_\agentsEC A \in \CS$,
    then
    $A \in \evidence_\agentsEC (w, c)$;
    \hfill $(\text{constant specification})$
\item
    if
    $(A \rightarrow B) \in \evidence_\agentsC (w, t)$
    and
    $A \in \evidence_\agentsC (w, s)$,
    then
    $B \in \evidence_\agentsC (w, t \cdot s)$;
    \hfill $(\text{application})$
\item
    $\evidence_\agentsC (w, s) \cup \evidence_\agentsC (w, t) \subseteq \evidence_\agentsC (w, s + t)$;
    \hfill $(\text{sum})$
\item
    if
    $A \in \evidence_\agents (w, t)$,
    then
    $\jbox{t}_\agents A \in \evidence_\agents (w, ! t)$;
    \hfill $(\text{inspection})$
\item
    if
    $A \in \evidence_i (w, t_i)$
    for all
    $1 \leq i \leq \numberofagents$,
    then
    $A \in \evidence_\ek (w, \tupling{t_1, \dots, t_\numberofagents})$;
    \hfill $(\text{tupling})$
\item
    if
    $A \in \evidence_\ek (w, t)$,
    then
    $A \in \evidence_i (w, \projection_i t)$;
    \hfill $(\text{projection})$
\item
    if
    $A \in \evidence_\ck (w, t)$,
    then
    $A \in \evidence_\ek (w, \cclhead{t})$
    and
    $\jbox{t}_\ck A \in \evidence_\ek (w, \ccltail{t})$;
    \hfill $(\text{co-closure})$
\item
    if
    $A \in \evidence_\ek (w, s)$
    and
    $(A \rightarrow \jbox{s}_\ek A) \in \evidence_\ck (w, t)$, \\
    then
    $A \in \evidence_\ck (w, \ind{t,s})$.
    \hfill $(\text{induction})$
\end{enumerate}
\end{definition}
When the model is clear from the context, we will directly refer to
$R_1, \dots, R_\numberofagents$,
$R_\ek$, $R_\ck$,
$\evidence_1, \dots, \evidence_\numberofagents$,
$\evidence_\ek$, $\evidence_\ck$, $W$, and~$\valuation$.
	
\begin{definition}
A ternary relation
$\mathcal{M}, w \Vdash A$
for \emph{formula~$A$ being satisfied at a world
$w \in W$
in an AF-model
$\mathcal{M} = (W, R, \evidence, \valuation)$}
is defined by induction on the structure of the formula~$A$:
\begin{enumerate}[leftmargin=*,topsep=0ex]
\item
    $\mathcal{M}, w \Vdash P$
    if and only if
    $w \in \valuation(P)$;
\item
    $\Vdash$~behaves classically with respect to the propositional connectives;
\item
    $\mathcal{M}, w \Vdash \jbox{t}_\agentsEC A$
    if and only if
    1)~$A \in \evidence_\agentsEC (w, t)$
    and
    2)~$\mathcal{M}, v \Vdash A$
    for all
    $v \in W$
    with
    $(w, v) \in R_\agentsEC$.
\end{enumerate}
We write
$\mathcal{M} \Vdash A$
if
$\mathcal{M}, w \Vdash A$
for all
$w \in W$.
We write\/
$\Vdash_\CS A$
and say that formula~$A$ is \emph{valid with respect to~$\CS$} if
$\mathcal{M} \Vdash A$
for all AF-models~$\mathcal{M}$ meeting~$\CS$.
\end{definition}
	
\begin{lemma}[Soundness]
\label{soundness}
Provable formulae are valid:\/
$\proves_\CS A$
implies\/
$\Vdash_\CS A$.
\end{lemma}

\begin{proof}
Let
$\mathcal{M} = (W, R, \evidence, \valuation)$
be an AF-model meeting~$\CS$ and let
$w \in W$.
We show soundness by induction on the derivation of~$A$.
The cases for propositional tautologies, for the application, sum, reflexivity, and inspection axioms, and for modus ponens rule are the same as for the single-agent case in~\cite{Fit05APAL} and are, therefore, omitted.
We show the remaining five cases:
\begin{description}[leftmargin=*,topsep=0ex]
\item[(tupling)]
    Assume
    $\mathcal{M}, w \Vdash \jbox{t_i}_i A$
    for all
    $1 \leq i \leq \numberofagents$.
    Then for all
    $1\leq i \leq \numberofagents$,
    we have
    1)~$\mathcal{M}, v \Vdash A$
    for all
    $v \in W$
    with
    $(w, v) \in R_i$
    and
    2)~$A \in \evidence_i (w, t_i)$.
    So, by the tupling closure condition,
    $A \in \evidence_\ek (w, \tupling{t_1, \dots, t_\numberofagents})$
    from~2).
    Since by definition
    $R_\ek = \bigcup_{i=1}^\numberofagents R_i$,
    it follows from~1) that
    $\mathcal{M}, v \Vdash A$
    for all
    $v \in W$
    with
    $(w, v)\in R_\ek$.
    Hence,
    $\mathcal{M}, w \Vdash \jbox{\tupling{t_1, \dots, t_\numberofagents}}_\ek A$.
\item[(projection)]
    Assume
    $\mathcal{M}, w \Vdash \jbox{t}_\ek A$.
    Then
    1)~$\mathcal{M}, v \Vdash A$
    for all
    $v \in W$
    with
    $(w, v) \in R_\ek$
    and
    2)~$A \in \evidence_\ek (w, t)$.
    By the projection closure condition, it follows from~2) that
    $A \in \evidence_i (w, \projection_i t)$.
    In addition, since
    $R_\ek = \bigcup_{i=1}^\numberofagents R_i$,
    we get
    $\mathcal{M}, v \Vdash A$
    for all
    $v \in W$
    with
    $(w, v) \in R_i$
    by~1).
    Thus,
    $\mathcal{M}, w \Vdash \jbox{\projection_i t}_i A$.
\item[(co-closure)]
    Assume
    $\mathcal{M}, w \Vdash \jbox{t}_\ck A$.
    Then
    1)~$\mathcal{M}, v \Vdash A$
    for all
    $v \in W$
    with
    $(w, v) \in R_\ck $
    and
    2)~$A \in \evidence_\ck (w, t)$.
    It follows from~1) that for all
    $v' \in W$
    with
    $(w, v') \in R_\ek$,
    we have
    $\mathcal{M}, v' \Vdash A$
    since
    $R_\ek \subseteq R_\ck$;
    also, due to the monotonicity closure condition,
    $\mathcal{M}, v' \Vdash \jbox{t}_\ck A$
    since
    $R_\ek \circ R_\ck \subseteq R_\ck$.
    From~2), by the co-closure closure condition,
    $A \in \evidence_\ek (w, \cclhead{t})$
    and
    $\jbox{t}_\ck A \in \evidence_\ek (w, \ccltail{t})$.
    Hence,
    $\mathcal{M}, w \Vdash \jbox{\cclhead{t}}_\ek A$
    and
    $\mathcal{M}, w \Vdash \jbox{\ccltail{t}}_\ek \jbox{t}_\ck A$.
\item[(induction)]
    Assume
    $\mathcal{M}, w \Vdash A$
    and
    $\mathcal{M}, w \Vdash \jbox{t}_\ck (A \rightarrow \jbox{s}_\ek A)$.
    From the second assumption and the reflexivity of~$R_\ck$, we get
    $\mathcal{M}, w \Vdash A \rightarrow \jbox{s}_\ek A$;
    thus,
    $\mathcal{M}, w \Vdash \jbox{s}_\ek A$
    by the first assumption.
    So
    $A \in \evidence_\ek (w, s)$
    and, by the second assumption,
    $A \rightarrow \jbox{s}_\ek A \in \evidence_\ck (w, t)$.
    By the induction closure condition, we have
    $A \in \evidence_\ck (w, \ind{t,s})$.
    To show
    $\mathcal{M}, v \Vdash A$
    for all
    $v \in W$
    with
    $(w, v) \in R_\ck$,
    we prove that
    $\mathcal{M}, v \Vdash A$
    for all
    $v \in W$
    with
    $(w, v) \in (R_\ek)^n$
    by induction on the positive integer~$n$.

    The \textbf{base case $n = 1$} immediately follows from
    $\mathcal{M}, w \Vdash \jbox{s}_\ek A$.

    \textbf{Induction step.}
    Let
    $(w, v^\prime) \in (R_\ek)^n$
    and
    $(v^\prime, v) \in R_\ek$
    for some
    $v, v^\prime \in W$.
    By induction hypothesis,
    $\mathcal{M}, v^\prime \Vdash A$.
    Since
    $\mathcal{M}, w \Vdash \jbox{t}_\ck (A \rightarrow \jbox{s}_\ek A)$,
    we get
    $\mathcal{M}, v^\prime \Vdash A \rightarrow \jbox{s}_\ek A$.
    Thus,
    $\mathcal{M}, v^\prime \Vdash \jbox{s}_\ek A$,
    which yields
    $\mathcal{M}, v \Vdash A$.

    Finally, we conclude that
    $\mathcal{M}, w \Vdash \jbox{\ind{t,s}}_\ck A$.
\item[(axiom necessitation)]
    Let $A$~be an axiom and $c$~be a proof constant such that
    $\jbox{c}_\agentsEC A \in \CS$.
    Since $A$~is an axiom,
    $\mathcal{M}, w \Vdash A$
    for all
    $w \in W$,
    as shown above.
    Since $\mathcal{M}$~is an AF-model meeting~$\CS$, we also have
    $A \in \evidence_\agentsEC (w, c)$
    for all
    $w \in W$
    by the constant specification closure condition.
    Thus,
    $\mathcal{M}, w \Vdash \jbox{c}_\agentsEC A$
    for all
    $w \in W$.
    \qed
\end{description}
\end{proof}

\begin{definition}
Let $\CS$~be a constant specification.
A set~$\Phi$ of formulae is called \emph{$\CS$-consistent} if
$\Phi \nvdash_\CS \phi$
for some formula~$\phi$.
A set~$\Phi$ is called \emph{maximal $\CS$-consistent} if it is $\CS$-consistent and has no $\CS$-consistent proper extensions.
\end{definition}
Whenever safe, we do not mention the constant specification and only talk about consistent and maximal consistent sets.
It can be easily shown that maximal consistent sets contain all axioms of~$\LPhEC$ and are closed under modus ponens.

\begin{definition}
For a set~$\Phi$ of formulae, we define
\begin{equation*}
    \Phi / \agentsEC \colonequals  \{ A \; : \; \text{there is a } t \in \eterms_\agentsEC \text{ such that } \jbox{t}_\agentsEC A \in \Phi \}
    \rlap{\enspace.}
\end{equation*}
\end{definition}
	
\begin{definition}
Let $\CS$~be a constant specification.
The \emph{canonical AF-model
$\mathcal{M} = (W, R, \evidence, \valuation)$
meeting~$\CS$} is defined as follows:
\begin{enumerate}[leftmargin=*,topsep=0.4ex,itemsep=0.4ex]
\item
    $W \colonequals \{ w \subseteq \formulaeLPhEC \; : \; w \text{ is a maximal $\CS$-consistent set} \}$;
\item
    $R_i \colonequals \{ (w,v) \in W \times W \; : \; w / i \subseteq v \}$;
\item
    $\evidence_\agentsEC (w, t) \colonequals \{ A \in \formulaeLPhEC \; : \; \jbox{t}_\agentsEC A \in w \}$;
\item
    $\valuation(P_n) \colonequals \{ w \in W \; : \; P_n \in w \}$.
\end{enumerate}
\end{definition}

\begin{lemma}
\label{canonicalmodelisfittingmodel}
Let $\CS$~be a constant specification.
The canonical AF-model meeting~$\CS$ is an AF-model meeting~$\CS$.
\end{lemma}
\begin{proof}
The proof of reflexivity and transitivity of each~$R_i$, as well as the argument for the constant specification, application, sum, and inspection closure conditions, is the same  as in the single-agent case (see~\cite{Fit05APAL}).
We show the remaining five closure conditions:
\begin{description}[leftmargin=*,topsep=0ex]
\item[(tupling)]
    Assume
    $A \in \evidence_\agents (w, t_\agents)$
    for all
    $1 \leq \agents \leq \numberofagents$.
    By definition of~$\evidence_i$, we have
    $\jbox{t_\agents}_\agents A \in w$
    for all
    $1 \leq \agents \leq \numberofagents$.
    Therefore, by the tupling axiom and maximal consistency,
    $\jbox{\tupling{t_1, \dots, t_\numberofagents}}_\ek A \in w$.
    Thus,
    $A \in \evidence_\ek (w, \tupling{t_1, \dots, t_\numberofagents})$.
\item[(projection)]
    Assume
    $A \in \evidence_\ek (w, t)$.
    Thus, we have
    $\jbox{t}_\ek A \in w$.
    Then, by the projection axiom and maximal consistency,
    $\jbox{\projection_\agents t}_\agents A \in w$,
    and thus
    $A \in \evidence_\agents (w, \projection_\agents t)$.
\item[(co-closure)]
    Assume
    $A \in \evidence_\ck (w, t)$.
    Thus,
    $\jbox{t}_\ck A \in w$,
    and, by the co-closure axioms and maximal consistency,
    $\jbox{\cclhead{t}}_\ek A \in w$
    and
    $\jbox{\ccltail{t}}_\ek \jbox{t}_\ck A \in w$.
    Hence,
    $A \in \evidence_\ek (w, \cclhead{t})$
    and
    $\jbox{t}_\ck A \in \evidence_\ek (w, \ccltail{t})$.
\item[(induction)]
    Assume
    $A \in \evidence_\ek (w, s)$
    and
    $(A \rightarrow \jbox{s}_\ek A) \in \evidence_\ck (w, t)$.
    Then we have
    $\jbox{s}_\ek A \in w$
    and
    $\jbox{t}_\ck (A \rightarrow \jbox{s}_\ek A) \in w$.
    From
    $\proves_\CS \jbox{s}_\ek A \rightarrow A$
    (Lemma~\ref{lem:basicproperties1}.\ref{lem:basicproperties1:Erefl}) and the induction axiom, it follows by maximal consistency that
    $A \in w$
    and
    $\jbox{\ind{t,s}}_\ck A \in w$.
    Therefore,
    $A \in \evidence_\ck (w, \ind{t,s})$.
\item[(monotonicity)]
    We show only the case of
    $\agentsC = \ck$
    since the other cases are the same as in~\cite{Fit05APAL}.
    It is sufficient to prove by induction on the positive integer~$n$ that
    \begin{equation}
    \label{ts:eq:mon:1}
        \text{if } \jbox{t}_\ck A \in w \text{ and } (w, v) \in (R_\ek)^n, \text{ then } \jbox{t}_C A \in v
        \rlap{\enspace.}
    \end{equation}

    \textbf{Base case $n = 1$.}
    Assume
    $(w, v) \in R_\ek$,
    i.e.,~$w / \agents \subseteq v$
    for some~$\agents$.
    As
    $\jbox{t}_\ck A \in w$,
    $\jbox{\projection_\agents \ccltail{t}}_\agents \jbox{t}_\ck A \in w$
    by maximal consistency, and hence
    $\jbox{t}_\ck  A \in w / \agents \subseteq v$.
    The argument for the \textbf{induction step} is similar.

    Now assume
    $(w, v) \in R_\ck = \bigcup_{n=1}^{\infty} (R_\ek)^n$
    and
    $A \in \evidence_\ck (w, t)$,
    i.e.,~$\jbox{t}_\ck A \in w$.
    As shown above,
    $\jbox{t}_\ck A \in v$.
    Thus,
    $A \in \evidence_\ck (v, t)$.
    \qed
\end{description}
\end{proof}

\begin{remark}
\label{nonstandardbehaviour}
Let $R_\ck^\prime$~denote the binary relation on~$W$ given by
\[
    (w, v) \in R_\ck^\prime \quad \text{if and only if} \quad w / \ck \subseteq v
    \rlap{\enspace.}
\]
An argument similar to the one just used for monotonicity shows that
$R_\ck \subseteq R_\ck^\prime$.
However, the converse does not hold for any pure $\ck$-axiomatically appropriate constant specification~$\CS$, which we demonstrate by adapting an example from~\cite{HM95}.
Let
\begin{equation*}
    \Phi \colonequals \{ \jbox{s_n}_\ek \dots \jbox{s_1}_\ek  P \; : \; n \geq 1, \, s_1, \dots, s_n \in \eterms_\ek \} \cup \{ \neg \jbox{t}_\ck P \; : \; t \in \eterms_\ck \}
    \rlap{\enspace.}
\end{equation*}
This set is $\CS$-consistent for any
$P \in \propositions$.

To see this, let
$\Phi^\prime \subseteq \Phi$
be finite and let $m$~denote the maximal number of terms such that
$\jbox{s_m}_\ek \dots \jbox{s_1}_\ek P \in \Phi^\prime$.
Define the model
$\mathcal{N} \colonequals (\N, R^\mathcal{N}, \evidence^\mathcal{N}, \valuation^\mathcal{N})$
by
\begin{itemize}[topsep=0ex,itemsep=0.5ex]
\item
    $R^\mathcal{N}_i \colonequals \{ (n, n+1) \in \N^2 \; : \; n \mod h = i \} \cup \{ (n, n) \in \N^2 \; : \; n \in \N \}$;
\item
    $\evidence^\mathcal{N} (n, s) \colonequals  \formulaeLPhEC$
    for all
    $n \in \N$
    and terms
    $s \in \eterms$;
\item
    $\valuation^\mathcal{N}(P) \colonequals  \{ 1, 2, \dots, m+1 \} \subseteq \N$.
\end{itemize}
Clearly, $\mathcal{N}$~meets any constant specification; in particular, it meets~$\CS$.
It can also be easily verified that
$\mathcal{N}, 1 \Vdash \Phi^\prime$;
therefore, $\Phi^\prime$~is $\CS$-consistent.

Since $\Phi$~is $\CS$-consistent, there exists a maximal $\CS$-consistent set
$w \supseteq \Phi$.
Let us show that the set
$\Psi \colonequals \{ \neg P \} \cup (w / \ck)$
is also $\CS$-consistent.
Indeed, if it were not the case, there would  exist formulae
$B_1, \dots, B_n \in w / \ck$ such that
\[
    \proves_\CS B_1 \rightarrow (B_2 \rightarrow \dots \rightarrow (B_n \rightarrow P) \dots )
    \rlap{\enspace.}
\]
Then, by Corollary~\ref{constructivenecessitation}, there would exist a term
$s \in \eterms_\ck$
such that
\[
    \proves_\CS \jbox{s}_\ck{(B_1 \rightarrow (B_2 \rightarrow \dots \rightarrow (B_n \rightarrow P) \dots ) )}
    \rlap{\enspace.}
\]
But this would imply
$\jbox{( \dots (s \cdot t_1) \cdots t_{n-1}) \cdot t_n}_\ck{P} \in w$
for
$\jbox{t_j}_\ck{B_j} \in w$,
$1 \leq j \leq n$,
a contradiction with the consistency of~$w$.

Let $v$~be a maximal $\CS$-consistent set that contains~$\Psi$,
i.e.,~$v \supseteq \Psi$.
Clearly,
$w / \ck \subseteq v$,
i.e.,~$(w, v) \in R_\ck^\prime$,
but
$(w, v) \notin R_\ck$
because this would imply
$P \in v$,
which cannot happen.
It follows that
$R_\ck \subsetneq R_\ck^\prime$.

Similarly, we can define~$R_\ek^\prime$ by
$(w, v) \in R_\ek^\prime$
if and only if
$w / \ek \subseteq v$.
However,
$ R_\ek^\prime = R_\ek$
for any $\ck$-axiomatically appropriate constant specification~$\CS$.
Indeed,  is easy to show that
$R_\ek \subseteq R_\ek^\prime$.
For the converse,  assume
$(w, v) \notin R_\ek$,
then
$(w, v) \notin R_i$
for all
$1 \leq i \leq h$.
So there are formulae
$A_1, \dots, A_\numberofagents$
such that
$\jbox{t_i}_i A_i \in w$
for some
$t_i \in \eterms_i$,
but
$A_i \notin v$.
Now let
$\jbox{c_i}_\ck (A_i \rightarrow A_1 \vee \dots \vee A_\numberofagents) \in \CS$
for constants
$c_1, \dots, c_\numberofagents$.
Then
$\jbox{\conversion_i c_i \cdot t_i}_i (A_1 \vee \dots \vee A_\numberofagents) \in w$
for all
$1 \leq i \leq h$,
so
$\jbox{\tupling{\conversion_1 c_1 \cdot t_1, \dots, \conversion_\numberofagents c_\numberofagents \cdot t_\numberofagents}}_\ek (A_1 \vee \dots \vee A_\numberofagents) \in w$.
However,
$A_i \notin v$
for any
$1 \leq i \leq h$;
therefore, by the maximal consistency of~$v$,
$A_1 \vee \dots \vee A_\numberofagents \notin v$
either.
Hence,
$w / \ek \nsubseteq v$,
so
$(w, v) \notin R_\ek^\prime$.
\end{remark}

\begin{lemma}[Truth Lemma]
\label{truthlemma}
Let $\CS$~be a constant specification and $\mathcal{M}$~be the canonical AF-model meeting~$\CS$.
For all formulae~$A$ and all worlds
$w \in W$,
\begin{equation*}
    A \in w \text{ if and only if } \mathcal{M}, w \Vdash A
    \rlap{\enspace.}
\end{equation*}
\end{lemma}

\begin{proof}
The proof is by induction on the structure of~$A$.
The cases for propositional variables and propositional connectives are immediate by the definition of~$\Vdash$ and by the maximal consistency of~$w$.
We check the remaining cases:

\noindent
\textbf{Case} $A$~is~$\jbox{t}_i B$.
Assume
$A \in w$.
Then
$B \in w / i$
and
$B \in \evidence_i(w,t)$.
Consider any~$v$ such that
$(w, v) \in R_i$.
Since
$w / i \subseteq v$,
it follows that
$B \in v$,
and thus, by induction hypothesis,
$\mathcal{M}, v \Vdash B$.
And
$\mathcal{M}, w \Vdash A$
immediately follows from this.
	
For the converse, assume
$\mathcal{M}, w \Vdash \jbox{t}_i B$.
By definition of~$\Vdash$ we get
$B \in \evidence_i (w, t)$,
from which
$\jbox{t}_i B \in w$
immediately follows by definition of~$\evidence_i$.

\noindent
\textbf{Case} $A$~is~$\jbox{t}_\ek B$.
Assume
$A \in w$
and consider any~$v$ such that
$(w, v) \in R_\ek$.
Then
$(w, v) \in R_i$
for some
$1 \leq i\leq h$,
i.e.,~$w / i \subseteq v$.
By definition of~$\evidence_\ek$, we get
$B \in \evidence_\ek (w, t)$.
By maximal consistency of~$w$, it follows that
$\jbox{\projection_i t}_i B \in w$,
and thus
$B \in w / i \subseteq v$.
Since, by induction hypothesis,
$\mathcal{M}, v \Vdash B$,
we conclude that
$\mathcal{M}, w \Vdash A$.
The argument for the converse repeats the one from the previous case.
	
\noindent
\textbf{Case} $A$~is~$\jbox{t}_\ck B$.
Assume
$A \in w$
and consider any~$v$ such that
$(w, v)\in R_\ck$,
i.e.,~$(w, v) \in (R_\ek)^n$
for some
$n \geq 1$.
As in the previous cases,
$B \in \evidence_\ck (w, t)$
by definition of~$\evidence_\ck$.
By~\eqref{ts:eq:mon:1} we find
$A \in v$,
and thus, by $\ck$-reflexivity and maximal consistency, also
$B \in v$.
Hence, by the induction hypothesis
$\mathcal{M}, v \Vdash B$.
Now
$\mathcal{M}, w \Vdash A$ immediately follows.
The argument for the converse repeats the one from the previous cases.
\qed
\end{proof}

Note that the converse directions in the proof above are far from trivial in the modal case, see e.g.~\cite{HM95}.
The last case, in particular, usually requires more sophisticated methods that guarantee the finiteness of the model.
	
\begin{theorem}[Completeness]
$\LPhEC(\CS)$~is sound and complete with respect to the class of AF-models meeting~$\CS$, i.e.,~for all formulae
$A \in \formulaeLPhEC$,
\begin{equation*}
    \proves_\CS A \text{ if and only if\/ } \Vdash_\CS A
    \rlap{\enspace.}
\end{equation*}
\end{theorem}

\begin{proof}
Soundness has already been shown in Lemma~\ref{soundness}.
For completeness, let $\mathcal{M}$~be the canonical AF-model meeting~$\CS$ and assume
$\nvdash_\CS A$.
Then $\{ \neg A \}$~is $\CS$-consistent and hence is contained in some maximal $\CS$-consistent set
$w \in W$.
So, by Lemma~\ref{truthlemma},
$\mathcal{M}, w \Vdash \neg A$,
and hence, by Lemma~\ref{canonicalmodelisfittingmodel},
$\nVdash_\CS A$.
\qed
\end{proof}
	
\noindent
M-models were introduced as semantics for~$\LP$ by Mkrtychev~\cite{Mkr97LFCS}.
They form a subclass of F-models (see~\cite{Fit05APAL}).
\begin{definition}
An \emph{M-model} is a singleton AF-model.
\end{definition}

\begin{theorem}[Completeness with respect to M-models]
$\LPhEC(\CS)$~is also sound and complete with respect to the class of M-models meeting~$\CS$.
\end{theorem}
\begin{proof}
Soundness follows immediately from Lemma~\ref{soundness}.
Now assume that
$\nvdash_\CS A$,
then $\{ \neg A \}$~is $\CS$-consistent, and hence
$\mathcal{M}, w \Vdash \neg A$
for some world
$w_0 \in W$
in the canonical AF-model
$\mathcal{M} = (W, R, \evidence, \valuation)$
meeting~$\CS$.

Let
$\mathcal{M}^\prime = (W^\prime, R^\prime, \evidence^\prime, \valuation^\prime)$
be the restriction of~$\mathcal{M}$ to~$\{ w_0 \}$,
i.e.,~$W^\prime \colonequals \{ w_0 \}$,
$R^\prime_\agentsEC \colonequals \{ (w_0, w_0) \}$
for any~$\agentsEC$,
$\evidence^\prime \colonequals \evidence \upharpoonright (W^\prime \times \eterms)$,
and
$\valuation^\prime(P_n) \colonequals \valuation(P_n) \cap W^\prime$.

Since $\mathcal{M}^\prime$~is clearly an M-model meeting~$\CS$, it remains to demonstrate that
$\mathcal{M}^\prime, w_0 \Vdash B$
if and only if
$\mathcal{M}, w_0 \Vdash B$
for all formulae~$B$.
We proceed by induction on the structure of~$B$.
The cases where  either $B$~is a propositional variable or its primary connective is propositional are trivial.
Therefore, we only show the case of
$B = \jbox{t}_\agentsEC C$.
First, observe that
\begin{equation}
\label{eq:mmodel:1}
    \mathcal{M}, w_0 \Vdash \jbox{t}_\agentsEC C \text{ if and only if } C \in \evidence^\prime_\agentsEC (w_0, t)
    \rlap{\enspace.}
\end{equation}
Indeed, by Lemma~\ref{truthlemma},
$\mathcal{M}, w_0 \Vdash \jbox{t}_\agentsEC C$
if and only if
$\jbox{t}_\agentsEC C \in w_0$,
which, by definition of the canonical AF-model, is equivalent to
$C \in \evidence_\agentsEC (w_0, t) = \evidence^\prime_\agentsEC (w_0, t)$.

If
$\mathcal{M}, w_0 \Vdash \jbox{t}_\agentsEC C$,
then
$\mathcal{M}, w_0 \Vdash C$
since $R_\agentsEC$~is reflexive.
By induction hypothesis,
$\mathcal{M}^\prime, w_0 \Vdash C$.
By~\eqref{eq:mmodel:1} we have
$C \in \evidence^\prime_\agentsEC (w_0, t)$,
and thus
$\mathcal{M}^\prime, w_0 \Vdash \jbox{t}_\agentsEC C$.

If
$\mathcal{M}, w_0 \nVdash \jbox{t}_\agentsEC C$,
then by~\eqref{eq:mmodel:1} we have
$C \notin \evidence^\prime_\agentsEC (w, t)$,
so
$\mathcal{M}^\prime, w_0 \nVdash \jbox{t}_\agentsEC C$.
\qed
\end{proof}

\begin{corollary}[Finite model property]
$\LPhEC(\CS)$~enjoys the finite model property with respect to AF-models.
\end{corollary}

\section{Conservativity}
\label{sect:conserv}
	
Yavorskaya in~\cite{TYav08TOCSnonote} introduced a two-agent version of~$\LP$, which we extend to an arbitrary~$h$ in the natural way:
\begin{definition}
The language of\/~$\LPh$ is obtained from that of\/~$\LPhEC$ by restricting the set of operations to~$\cdot_i$, $+_i$, and\/~$!_i$ and by dropping all terms from\/~$\eterms_\ek$ and\/~$\eterms_\ck$.
The axioms are restricted to application, sum, reflexivity, and inspection for each~$i$.
The definition of constant specification is changed accordingly.
\end{definition}
We show that $\LPhEC$~is conservative over~$\LPh$ by adapting a technique from~\cite{Fit08AMAI}.
	
\begin{definition}
The mapping\/
$\times : \formulaeLPhEC \to \formulaeLPh$
is defined as follows:
\begin{enumerate}[leftmargin=*,topsep=0ex]
\item
    $P^\times \colonequals P$
    for propositional variables
    $P \in \propositions$;
\item
    $\times$~commutes with propositional connectives;
\item
    $(\jbox{t}_\agentsEC A)^\times \colonequals
    \begin{cases}
        A^\times                    & \text{if $t$~contains a subterm } s \in \eterms_\ek \cup \eterms_\ck, \\
        \jbox{t}_\agentsEC A^\times & \text{otherwise.}
    \end{cases}$
\end{enumerate}
\end{definition}
	
\begin{theorem}
Let $\CS$~be a constant specification for\/~$\LPhEC$.
For an arbitrary formula
$A \in \formulaeLPh$,
if\/
$\LPhEC(\CS) \proves A$,
then\/
$\LPh(\CS^\times) \proves A$.
\end{theorem}

\begin{proof}	
Since
$A^\times = A$
for any
$A \in \formulaeLPh$,
it suffices to demonstrate that for any formula
$D \in \formulaeLPhEC$,
if
$\LPhEC(\CS) \proves D$,
then
$\LPh(\CS^\times) \proves D^\times$,
which can be done by induction on the derivation of~$D$.

\Case when $D$~is a propositional tautology, then so is~$D^\times$.

\Case when
$D = \jbox{t}_\agents B \rightarrow B$
is an instance of the reflexivity axiom.
Then $D^\times$~is either
$\jbox{t}_\agents B^\times \rightarrow B^\times$
or
$B^\times \rightarrow B^\times$,
i.e.,~an instance of the reflexivity axiom of~$\LPh$ or a propositional tautology respectively.

\Case when
$D = \jbox{t}_\agentsC (B \rightarrow C) \rightarrow (\jbox{s}_\agentsC B \rightarrow \jbox{t \cdot s}_\agentsC C)$
is an instance of the application axiom.
We distinguish the following possibilities:
\begin{enumerate}[leftmargin=*,topsep=0ex]
\item
    Both~$t$ and~$s$ contain a subterm from
    $\eterms_\ek \cup \eterms_\ck$.
    Then $D^\times$~has the form
    $(B^\times \rightarrow C^\times) \rightarrow (B^\times \rightarrow C^\times)$,
    which is a propositional tautology and, thus, an axiom of~$\LPh$.
\item
    Neither~$t$ nor~$s$ contains a subterm from
    $\eterms_\ek \cup \eterms_\ck$.
    Then $D^\times$~is an instance of the application axiom of~$\LPh$.
\item
    Term~$t$ contains a subterm from
    $\eterms_\ek \cup \eterms_\ck$
    while $s$~does not.
    Then $D^\times$~is
    $(B^\times \rightarrow C^\times) \rightarrow (\jbox{s}_\agents B^\times \rightarrow C^\times)$,
    which can be derived in~$\LPh(\CS^\times)$ from the reflexivity axiom
    $\jbox{s}_\agents B^\times \rightarrow B^\times$
    by propositional reasoning.
    In this case, translation~$\times$ does not map an axiom of~$\LPhEC$ to an axiom of~$\LPh$.
\item
    Term~$s$ contains a subterm from
    $\eterms_\ek \cup \eterms_\ck$
    while $t$~does not.
    Then $D^\times$~is
    $\jbox{t}_\agents (B^\times \rightarrow C^\times) \rightarrow (B^\times \rightarrow C^\times)$,
    an instance of the reflexivity axiom of~$\LPh$.
\end{enumerate}
		
\Case when
$D = \jbox{t}_\agentsC B \rightarrow \jbox{t + s}_\agentsC B$
is an instance of the sum axiom.
Then $D^\times$~becomes
$B^\times \rightarrow B^\times$,
$\jbox{t}_\agents B^\times \rightarrow B^\times$,
or
$\jbox{t}_\agents B^\times \rightarrow \jbox{t + s}_\agents B^\times$,
i.e.,~a propositional tautology, an instance of the reflexivity axiom of~$\LPh$, or an instance of the sum axiom of~$\LPh$ respectively.
The sum axiom
$\jbox{s}_\agentsC B \rightarrow \jbox{t + s}_\agentsC B$ is treated in the same manner.
		
\Case when
$D = \jbox{t}_\agents B \rightarrow \jbox{! t}_\agents \jbox{t}_\agents B$
is an instance of the inspection axiom.
Then $D^\times$~is either the propositional tautology
$B^\times \rightarrow B^\times$
or
$\jbox{t}_\agents B^\times \rightarrow \jbox{! t}_\agents \jbox{t}_\agents B^\times$,
an instance of the inspection axiom of~$\LPh$.

\Case when
$D = \jbox{t_1}_1 B \wedge \dots \wedge \jbox{t_\numberofagents}_\numberofagents B \rightarrow \jbox{\tupling{t_1, \dots, t_\numberofagents}}_\ek B$
is an instance of the tupling axiom.
We distinguish the following possibilities:
\begin{enumerate}[leftmargin=*,topsep=0ex]
\item
    At least one of the~$t_i$'s contains a subterm from
    $\eterms_\ek \cup \eterms_\ck$.
    Then $D^\times$~has the form
    $C_1 \wedge \dots \wedge C_\numberofagents \rightarrow B^\times$
    with at least one
    $C_i = B^\times$
    and is, therefore, a propositional tautology.
\item
    None of the~$t_i$'s contains a subterm from
    $\eterms_\ek \cup \eterms_\ck$.
    Then $D^\times$~has the form
    $\jbox{t_1}_1 B^\times \wedge \dots \wedge \jbox{t_\numberofagents}_\numberofagents B^\times \rightarrow  B^\times$,
    which can be derived in~$\LPh(\CS^\times)$ from the reflexivity axiom.
    This is another case when translation~$\times$ does not map an axiom of~$\LPhEC$ to an axiom of~$\LPh$.
\end{enumerate}

\Case when $D$~is an instance of the projection axiom
$\jbox{t}_\ek B \rightarrow \jbox{\projection_i t}_i B$
or of the co-closure axiom,
i.e.,~$\jbox{t}_\ck B \rightarrow \jbox{\cclhead{t}}_\ek B$
or
$\jbox{t}_\ck B \rightarrow \jbox{\ccltail{t}}_\ek \jbox{t}_\ck B$.
Then $D^\times$~is the propositional tautology
$B^\times \rightarrow B^\times$.
	
\Case when
$D = B \wedge \jbox{t}_\ck (B \rightarrow \jbox{s}_\ek B) \rightarrow \jbox{\ind{t,s}}_\ck B$
is an instance of the induction axiom.
Then $D^\times$~is
$B^\times \wedge (B^\times \rightarrow B^\times) \rightarrow B^\times$, a propositional tautology.

\Case when $D$~is derived by modus ponens is trivial.
	
\Case when $D$~is
$\jbox{c}_\agentsEC B \in \CS$.
Then $D^\times$~is either~$B^\times$
or~$\jbox{c}_\agents B^\times$.
In the former case, $B$~is an axiom of~$\LPhEC$, and hence $B^\times$~is derivable in~$\LPh(\CS^\times)$, as shown above; in the latter case,
$\jbox{c}_\agents B^\times \in \CS^\times$.
\qed
\end{proof}

\begin{remark}
Note that $\CS^\times$~need not, in general, be a constant specification for~$\LPh$ because, as noted above, for an axiom~$D$ of~$\LPhEC$, its image~$D^\times$ is not always an axiom of~$\LPh$.
To ensure that $\CS^\times$~is a proper constant specification,
$(A \rightarrow B) \rightarrow (\jbox{s}_\agents A \rightarrow B)$
and
$\jbox{t_1}_1 A \wedge \dots \wedge \jbox{t_\numberofagents}_\numberofagents A \rightarrow A$
have to be made axioms of~$\LPh$.
Another option is to use Fitting's concept of \emph{embedding} one justification logic into another, which involves replacing constants in~$D$ with more complicated terms in~$D^\times$ (see~\cite{Fit08AMAI} for details).
\end{remark}
	
\section{Forgetful Projection and a Word on Realization}
\label{sect:real}
	
Most justification logics are introduced as explicit counterparts to particular modal logics in the strict sense described in Sect.~\ref{sec:intro}.
Although the realization theorem for~$\LPhEC$ remains an open problem, in this section we prove that each theorem of our logic~$\LPhEC$ states a valid modal fact if all terms are replaced with the corresponding modalities, which is one direction of the realization theorem.
We also discuss approaches to the harder opposite direction.

We start with recalling the modal language of common knowledge.
Modal formulae are defined by the following grammar
\begin{equation*}
    A \coloncolonequals P_j \dashspace \neg A \dashspace (A \wedge A) \dashspace (A \vee A) \dashspace (A \rightarrow A) \dashspace \Box_i A \dashspace \ek A \dashspace \ck A
    \rlap{\enspace,}
\end{equation*}
where $P_j \in \propositions$.
The set of all modal formulae is denoted by~$\formulaeSFourhEC$.
	
The Hilbert system~$\SFourhEC$~\cite{HM95} is given by the modal axioms of~$\SFour$ for individual agents, by the necessitation rule for
$\Box_1, \dots, \Box_\numberofagents$,
and~$\ck$, by modus ponens, and by the axioms
\begin{gather*}
    \ck (A \rightarrow B) \rightarrow (\ck A \rightarrow \ck B),
    \qquad
    \ck A \rightarrow A,
    \qquad
    \ek A \leftrightarrow \Box_1 A \wedge \dots \wedge \Box_\numberofagents A, \\
	A \wedge \ck (A \rightarrow \ek A) \rightarrow \ck A,
    \qquad\qquad
    \ck A \rightarrow \ek (A \wedge \ck A).
\end{gather*}
	
\begin{definition}[Forgetful projection]
\label{def:forgetfulprojection}
The mapping\/
$\forgetful \colon \formulaeLPhEC \rightarrow \formulaeSFourhEC$
is defined as follows:
\begin{enumerate}[topsep=0ex]
\item
    $P^\forgetful \colonequals P$
    for propositional variables
    $P \in \propositions$;
\item
    $\forgetful$~commutes with propositional connectives;
\item
    $(\jbox{t}_i A)^\forgetful \colonequals \Box_i A^\forgetful$;
\item
    $(\jbox{t}_\ek A)^\forgetful \colonequals \ek A^\forgetful$;
\item
    $(\jbox{t}_\ck A)^\forgetful \colonequals \ck A^\forgetful$.
\end{enumerate}
\end{definition}
		
\begin{lemma}
\label{l:fp:1}
Let $\CS$~be any constant specification.
For any formula
$A \in \formulaeLPhEC$,
if\/
$\LPhEC(\CS) \proves A$,
then\/
$\SFourhEC \proves A^\forgetful$.
\end{lemma}

\begin{proof}
The proof is by easy induction on the derivation of~$A$.
\qed
\end{proof}

\begin{definition}[Realization]
A realization is a mapping
$r \colon \formulaeSFourhEC \to \formulaeLPhEC$
such that\/
$(r(A))^\forgetful = A$.
We usually write~$A^r$ instead of\/~$r(A)$.
\end{definition}
	
We can think of a realization as a function that replaces occurrences of modal operators (including~$\ek$ and~$\ck$) with evidence terms of the corresponding type.
The problem of realization for a given pure $\ck$-axiomatically appropriate constant specification~$\CS$ can be stated as follows:
\[
    \text{Is there a realization~$r$ such that
    $\LPhEC(\CS) \proves A^r$
    for any theorem~$A$ of~$\SFourhEC$?}
\]
A positive answer to this question would constitute the harder direction of the realization theorem, which is often demonstrated using induction on a cut-free sequent proof of the modal formula.

Cut-free systems  for~$\SFourhEC$  are presented in~\cite{aj05} and~\cite{bs09}.
They are based on an infinitary $\omega$-rule of the form
\[
    \hilbertrule{\ek^m A, \Gamma \quad \text{for all } m \geq 1}{\ck A, \Gamma} \qquad (\omega).
\]
However, realization of such a rule meets with serious difficulties in reaching uniformity among the realizations of the approximants~$\ek^m A$.

A finitary cut-free system is obtained in~\cite{jks07} by finitizing this $\omega$-rule via the finite model property.
Unfortunately, the ``somewhat unusual'' structural properties of the resulting system (see discussion in~\cite{jks07}) make it hard to use it for realization.

The non-constructive, semantic realization method from~\cite{Fit05APAL} cannot be applied directly because of the non-standard behavior of the canonical model (see Remark~\ref{nonstandardbehaviour}).

Perhaps the infinitary system  presented in~\cite{BucKuzStu09M4M}, which is finitely branching but admits infinite branches, can help in proving the realization theorem for~$\LPhEC$.
For now this remains work in progress.

\section{Coordinated attack}
\label{s:attack:1}

To illustrate our logic, we will now analyze the coordinated attack problem along the lines of~\cite{FHMV95}, where additional references can be found.
Let us briefly recall this classical problem.
Suppose two divisions of an army, located in different places, are about to attack an enemy.
They have some means of communication, but these may be unreliable, and the only way to secure a victory is to attack simultaneously.
How should generals~$G$ and~$H$ who command the two divisions coordinate their attacks?
Of course, general~$G$ could send a message~$m_1^G$ with the time of attack to general~$H$.
Let us use the proposition~$\tsd$ to denote the fact that the message with the time of attack has been delivered.
If the generals trust the authenticity of the message, say because of a signature, the message itself can be taken as evidence that it has been delivered.
So general~$H$, upon receiving the message, knows the time of attack,
i.e.,~$\jbox{m_1^G}_H \tsd$.
However, since communication is unreliable, $G$~considers it possible that his message has not been delivered.
But if general~$H$ sends an acknowledgment~$m_2^H$, he in turn cannot be sure whether the acknowledgment has reached~$G$, which prompts yet another acknowledgment~$m_3^G$ by general~$G$, and so on.

In fact, common knowledge of~$\tsd$ is a necessary condition for the attack.
Indeed, it is reasonable to assume it to be common knowledge between the generals that they should only attack simultaneously or not attack at all, i.e.,~that they attack only if both know that they attack:
$\jbox{t}_\ck (\attack \rightarrow \jbox{s}_\ek \attack)$
for some terms~$s$ and~$t$.
Thus, by the induction axiom, we get
$\attack \rightarrow \jbox{\ind{t,s}}_\ck \attack$.
Another reasonable assumption is that it is common knowledge that neither general attacks unless the message with the time of attack has been delivered:
$\jbox{r}_\ck (\attack \rightarrow \delivered)$
for some term~$r$.
Using the application axiom, we obtain
$\attack \rightarrow \jbox{r \cdot \ind{t,s}}_\ck \delivered$.

We now show that common knowledge of~$\tsd$ cannot be achieved and that, therefore, no attack will take place, no matter how many messages and acknowledgments~$m_1^G$, $m_2^H$, $m_3^G$,~\dots are sent by the generals even if all the messages are successfully delivered.

In the classical modeling without evidence, the reason is that the sender of the last message always considers the possibility that his last message, say~$m_{2k}^H$, has not been delivered.
To give a flavor of the argument carried out in detail in~\cite{FHMV95}, we provide a countermodel where $m_2^H$~is the last message, it has been delivered, but $H$~is unsure of that,
i.e.,~$\jbox{m_1^G}_H {\tsd}$,
$\jbox{m_2^H}_G {\jbox{m_1^G}_H {\tsd}}$,
but
$\lnot \jbox{s}_H \jbox{m_2^H}_G {\jbox{m_1^G}_H {\tsd}}$
for all terms~$s$.
Indeed, consider the model~$\tsm$ with
$W \colonequals \{ 0, 1, 2, 3 \}$,
$\tsval(\tsd) \colonequals \{ 0, 1, 2 \}$,
$R_G$~being the reflexive closure of~$\{ (1, 2) \}$, $R_H$~being the reflexive closure of~$\{ (0, 1), (2, 3) \}$, and any evidence function~$\evidence$ such that
$\tsd \in \evidence_H (0, m_{1}^G)$
and
$\jbox{m_1^G}_H {\tsd} \in \evidence_G (0, m_{2}^H)$.
Then, whatever $\evidence_\ck$~is, we have
$\mathcal{M}, 0 \nVdash \jbox{s}_H \jbox{m_2^H}_G {\jbox{m_1^G}_H {\tsd}}$
and
$\mathcal{M}, 0 \nVdash \jbox{t}_\ck \tsd$
for any~$s$ and~$t$ because
$\mathcal{M}, 3 \nVdash \tsd$.

In our models with explicit evidence, there is an alternative possibility for the lack of knowledge: the absence of evidence.
For example, $G$~may receive the acknowledgment~$m_2^H$ but not consider it to be evidence for
$\jbox{m_1^G}_H \tsd$
because the signature of~$H$ is missing.

We now demonstrate that common knowledge of the time of attack cannot emerge, basing the argument solely on the lack of common knowledge evidence.
A corresponding M-model
$\tsm = (W, R, \evidence, \valuation)$
is obtained as follows:
$W \colonequals \{ w \}$,
$R_i \colonequals \{ (w, w) \}$,
$\tsval(\tsd) \colonequals \{ w \}$,
and $\tse$~is the minimal evidence function such that
$\tsd \in \evidence_H (w, m_{1}^G)$
and
$\jbox{m_1^G}_H{\tsd} \in \evidence_G (w, m_{2}^H)$.
In this model
$M, w \nVdash \jbox{t}_\ck \tsd$
for any evidence term~$t$ because
$\tsd \notin \evidence_\ck (w, t)$
for any~$t$.
To show the latter statement, note that for any term~$t$, by Lemma~\ref{l:fp:1},
\begin{equation}
\label{eq:coordatt}
    \nvdash \jbox{m_1^G}_H \tsd \wedge \jbox{m_2^H}_G \jbox{m_1^G}_H{\tsd} \rightarrow \jbox{t}_\ck \tsd
\end{equation}
because
\[
    \SFourhEC \nvdash \Box_H \tsd \wedge \Box_G \Box_H \tsd \rightarrow \ck \tsd
    \rlap{\enspace,}
\]
which is easy to demonstrate.
Thus, the negation of the formula from~\eqref{eq:coordatt} is satisfiable, and for each~$t$ there is a world~$w_t$ in the canonical AF-model with evidence function~$\tsecan$ such that
$\tsd \in \tsecan_H (w_t, m_1^G)$
and
$\jbox{m_1^G}_H \tsd \in \tsecan_G (w_t, m_2^H)$,
but by the Truth Lemma~\ref{truthlemma},
$\tsd \notin \tsecan_\ck (w_t, t)$.
Since
$\tsecan \upharpoonright (\{ w_t \} \times \eterms)$
satisfies all the closure conditions, minimality of~$\tse$ implies that
$\tse_\ck (w, s) \subseteq \tsecan_\ck (w_t, s)$
for any term~$s$.
In particular,
$\tsd \notin \evidence_\ck (w, t)$
for any term~$t$.

\section{Conclusions}
\label{sect:concl}

We have presented an explicit evidence system~$\LPhEC$ with common knowledge, which is a conservative extension of the multi-agent explicit evidence logic~$\LPh$.
The major open problem  at the moment remains proving the realization theorem, one direction of which we have demonstrated.

Our analysis of the coordinated attack problem in the language of~$\LPhEC$ shows that access to explicit evidence creates more alternatives than the classical modal approach.
In particular, the lack of knowledge can  occur either because messages are not delivered or because evidence of authenticity is missing.

We have mostly concentrated on the study of $\ck$-axiomatically appropriate constant specifications.
For modeling distributed systems with different reasoning capabilities of agents, it is also interesting to consider $\agents$-axiomatic appropriate, $\ek$-axiomatic appropriate, and mixed constant specifications, where only certain aspects of reasoning are common knowledge.

We established soundness and completeness with respect to AF-models and singleton M-models.
Can other semantics for justification logics such as (arithmetical) provability semantics~\cite{Art95TR,Art01BSL} and game semantics~\cite{Ren09ICnonote} be adapted to~$\LPhEC$?

There are further interesting questions: Is $\LPhEC$~decidable and, if yes, what is its complexity compared to that of~$\SFourhEC$?
How robust is our treatment of common knowledge if the individual modalities are taken to be of type~$\textsf{K}$, $\textsf{K5}$,~etc.?

\phantomsection
\addcontentsline{toc}{section}{References}
\bibliographystyle{alphaurl}
\bibliography{bibliography,JLBibliography}

\end{document}